\documentclass[11pt]{article}
\usepackage[bottom=1in,left=1in, right=1in, top=1in]{geometry}
\usepackage{microtype}
\usepackage{amsmath}
\usepackage{amsthm}
\newtheorem{proposition}{Proposition}
\newtheorem{theorem}{Theorem}
\usepackage{amsfonts}
\usepackage{graphicx}
\usepackage{epstopdf}
\usepackage{tikz}
\usepackage{ifthen}
\usepackage{xcolor}
\usepackage{tcolorbox}
\usepackage{framed}
\usepackage{xifthen}
\usepackage{tabularx}
\usepackage{todonotes}  
\usepackage{hyperref}
\usepackage[capitalize]{cleveref}
\usepackage{amsopn}
\usepackage{derivative}
\usepackage{mathtools}

\newcommand{\pname}[1]{\textsc{#1}}
\newcommand{\cclass}[1]{\textsf{#1}}

\newcommand{\MinCost}{\pname{Structural Optimal Jacobian Accumulation}}
\newcommand{\MinRep}{\pname{Minimum Edge Count}}

\newcommand{\NP}{\cclass{\textup{NP}}}

\newcommand{\innbrs}[2]{N^-_{{#1}}({#2})}
\newcommand{\outnbrs}[2]{N^+_{{#1}}({#2})}

\newcommand{\gElimVertex}[2]{{#1}_{/{#2}}}
\newcommand{\gElimSet}[2]{{#1}_{#2}}
\newcommand{\gElimSequence}[2]{\gElimSet{#1}{#2}}

\title{Structural Optimal Jacobian Accumulation and Minimum Edge Count are NP-Complete Under Vertex Elimination}
\usepackage{authblk}
\author[1]{Matthias Bentert}
\author[2]{Alex Crane}
\author[1]{Pål Grønås Drange}
\author[2]{Yosuke Mizutani}
\author[2]{Blair~D.~Sullivan}
\affil[1]{University of Bergen, Norway}
\affil[2]{University of Utah, USA}
\date{}

\usetikzlibrary{positioning,backgrounds,patterns,calc,matrix,shapes,decorations.pathreplacing,decorations.pathmorphing,math}
\tikzset{crossing/.style={cross out, draw=red, minimum size=2*(#1-\pgflinewidth), inner sep=0pt, outer sep=1pt, very thick}, crossing/.default={4pt}}
\usetikzlibrary{arrows.meta,bending,decorations.markings,hobby,patterns,calc}
\tikzset{>={Latex[width=2mm,length=3mm]}}
\definecolor{cbred}{RGB}{200, 50, 80}
\definecolor{cbgreen}{RGB}{0, 128, 128}
\definecolor{cbblue}{RGB}{0, 76, 153}
\definecolor{cbyellow}{RGB}{204, 153, 0}

\ifpdf
  \DeclareGraphicsExtensions{.eps,.pdf,.png,.jpg}
\else
  \DeclareGraphicsExtensions{.eps}
\fi

\newlength{\RoundedBoxWidth}
\newsavebox{\GrayRoundedBox}
\newenvironment{GrayBox}[1]%
   {\setlength{\RoundedBoxWidth}{.93\textwidth}
    \def\boxheading{#1}
    \begin{lrbox}{\GrayRoundedBox}
       \begin{minipage}{\RoundedBoxWidth}}%
   {   \end{minipage}
    \end{lrbox}
    \begin{center}
    \begin{tikzpicture}%
       \node(Text)[draw=black!20,fill=white,rounded corners,%
             inner sep=2ex,text width=\RoundedBoxWidth]%
             {\usebox{\GrayRoundedBox}};
        \coordinate(x) at (current bounding box.north west);
        \node [draw=white,rectangle,inner sep=3pt,anchor=north west,fill=white]
        at ($(x)+(6pt,.75em)$) {\boxheading};
    \end{tikzpicture}
    \end{center}}

\newenvironment{defproblemx}[2][]{\noindent\ignorespaces%
                                \FrameSep=6pt%
                                \parindent=0pt%
                \vspace*{-1.5em}
                \ifthenelse{\isempty{#1}}{%
                  \begin{GrayBox}{\textsc{#2}}%
                }{%
                  \begin{GrayBox}{\textsc{#2}  parameterized by~{#1}}%
                }
                \begin{tabular*}{\textwidth}{@{\hspace{.1em}} >{\itshape} p{1.8cm} p{0.8\textwidth} @{}}%
            }{
                \end{tabular*}%
                \end{GrayBox}%
                \ignorespacesafterend
            }

\newcommand{\defproblem}[3]{
  \begin{defproblemx}{#1}
    Input:  & #2 \\
    Question: & #3
  \end{defproblemx}
}%

\begin{document}

\maketitle

\begin{abstract}
We study graph-theoretic formulations of two fundamental problems in algorithmic differentiation.
The first (\MinCost) is that of computing a Jacobian while minimizing multiplications. The second (\MinRep) is to find a minimum-size computational graph.
For both problems, we consider the vertex elimination operation.
Our main contribution is to show that both problems are \cclass{NP}-complete, thus resolving longstanding open questions.
In contrast to prior work, our reduction for~\MinCost{} does not rely on any assumptions about the algebraic relationships between local partial derivatives; we allow these values to be mutually independent.
We also provide $O^*(2^n)$-time exact algorithms for both problems, and show that under the exponential time hypothesis these running times are essentially tight.
Finally, we provide a data reduction rule for \MinCost{} by showing that false twins may always be eliminated consecutively.
\end{abstract}

\section{Introduction}

A core subroutine in numerous scientific computing applications is the efficient computation of derivatives.
If inexactness is acceptable, finite difference methods~\cite{quarteroni2006numerical} may be applicable.
Meanwhile, computer algebra packages allow for exact symbolic computation, though
these methods suffer from the requirement of a closed-form expression of the function $F$ to be differentiated as well as poor running time in practice.
A third approach is \emph{algorithmic differentiation}, also sometimes called \emph{automatic differentiation}. Algorithmic differentiation provides almost exact computations, incurring only rounding errors (in contrast to finite difference methods which also incur truncation errors). Furthermore, the running time required by algorithmic differentiation is bounded by the time taken to compute $F$ (in contrast to symbolic methods).
We refer to the works by Griewank and Walther~\cite{griewank2008evaluating} and by Naumann~\cite{naumann2011art} for an introduction and numerous applications.

In algorithmic differentiation, we assume that the function $F$ is implemented as a numerical program.
The key insight is that such programs consist of compositions of elementary functions, e.g., multiplication, \emph{sin}, \emph{cos}, etc., for which the derivatives are known or easily computable.
Then, derivative computations for the function $F$ follow by application of the chain rule.
The relevant numerical program can be modeled as a directed acyclic graph~(DAG)~$D = (S \uplus I, E)$, referred to as the \emph{computational graph} of $F$~\cite{griewank2005analysis}.
The \emph{source} and \emph{sink} vertices $S$ model the inputs and outputs of $F$, respectively. The \emph{internal} vertices~$I$ model elementary function calls. Arcs (directed edges) model data dependencies.
Let~$\mathcal P_{s, t}$ denote the set of all paths from source $s$ to sink $t$. If we associate to each arc $(u,v)$ the known (or easily computable) local partial derivative~$\pdv{v}{u}$,
then the chain rule allows us to compute the derivative of $t$ with respect to $s$ by
\begin{align}\label{eq:baur-chain-rule}
  \tag{1}
  \odv{t}{s} = \sum_{P \in \mathcal P_{s,t}} \prod_{(u, v) \in P} \pdv{v}{u},
\end{align}
as shown by Bauer~\cite{bauer1974computational}.
This weighted DAG is the \emph{linearized} computational graph.

Using the linearized computational graph, we may model the computation of a Jacobian as follows.
An \emph{elimination} of an internal vertex $v$ is the deletion of $v$ and the creation of all arcs from each in-neighbor to each out-neighbor of $v$ which are not already present.
We write~$\gElimVertex{D}{v}$ for the resulting DAG.
An \emph{elimination sequence}~${\sigma = (v_1, v_2, \ldots, v_\ell)}$ of length~$\ell$ is a tuple of internal vertices, and we denote
by~$\gElimSequence{D}{\sigma} = (((\gElimVertex{D}{v_1})_{/v_2})\ldots)_{/v_\ell}$ the result of eliminating these vertices in the order given by~$\sigma$.
We call $\sigma$ a \emph{total} elimination sequence if $\ell = |I|$.
If $\sigma$ is a total elimination sequence, then with appropriate local partial derivative computations or updates during the sequence, $\gElimSequence{D}{\sigma}$ can be thought of as a bipartite DAG (with sources on one side and sinks on the other) representing the Jacobian matrix of the associated numerical program.
To reflect the number of multiplications needed to maintain correctness of Equation~(\ref{eq:baur-chain-rule}),
we say that the \emph{cost} of eliminating an internal vertex~$v$ is the \emph{Markowitz degree}~$\mu(v)$ of $v$, that is, the in-degree of~$v$ times the out-degree of~$v$.
The cost of an elimination sequence is the sum of the costs of the involved eliminations.
We can now phrase the problem of computing a Jacobian with few multiplications in purely graph-theoretic terms:\looseness=-1

\bigskip
\defproblem{\MinCost{}}
{A DAG $D = (S \uplus I, A)$ and an integer~$k$.}
{Does there exist a total elimination sequence of cost at most $k$?}

\textsc{Optimal Jacobian Accumulation} is known to be \cclass{NP}-complete only under certain assumptions regarding algebraic dependencies between local partial derivatives~\cite{naumann2008optimal}.
Despite this, heuristics used in practice are based on the purely structural formulation presented here~\cite{chen2012integer,forth2004jacobian,griewank2003accumulating,pryce2008fast} and understanding the complexity of this formulation, open since at least 1993~\cite{griewank1993some}, has recently been
highlighted as an important problem in applied combinatorics~\cite{aksoy2023seven}.

\bigskip
A solution to \MinCost{} is always a total elimination sequence, resulting in a bipartite DAG representing the Jacobian.
A related problem, first posed by
Griewank and Vogel in 2005~\cite{griewank2005analysis}, is to identify a (not necessarily total) elimination sequence which results in a computational graph with a minimum number of arcs:

\bigskip
\defproblem{\MinRep{}}
{A DAG $D = (S \uplus I, A)$ and an integer~$k$.}
{Does there exist an elimination sequence~$\sigma$ (of any length) such that~$D_{\sigma}$ contains at most~$k$ arcs?}

The motivation to solve this problem is twofold. First, suppose that our function~$F$ is a map from $\mathbb{R}^n$ to~$\mathbb{R}^m$ and we wish to multiply its Jacobian (a matrix in $\mathbb{R}^{m \times n}$) by a matrix $S \in \mathbb{R}^{n \times q}$.
We may model this computation by augmenting the linearized computational graph $D$ with $q$ new source vertices. For each new source $i \in \{1, 2, \ldots, q\}$ and each original source $j \in \{1, 2, \ldots, n\}$, we add the arc $(i, j)$.
By labeling the new arcs with entries from the matrix $S$, we may obtain the result of the matrix multiplication via application of Equation~(\ref{eq:baur-chain-rule}).
We refer to the work by Mosenkis and Naumann~\cite{mosenkis2012optimality} for a formal presentation.
The number of multiplications required to use Equation~(\ref{eq:baur-chain-rule}) in this way grows with the number of arcs in~$D$, thus motivating the computation of a small (linearized) computational graph.
A second motivation is that~$D$ can sometimes reveal useful information not evident in the Jacobian matrix.
This situation, known as \emph{scarcity}, is described by Griewank and Vogel~\cite{griewank2005analysis}.
Thus, it is desirable to store~$D$, rather than the Jacobian matrix, and consequently it is also desirable for $D$ to be as small as possible.

Despite these motivations and several algorithmic studies~\cite{griewank2005analysis,lyons2008practical,mosenkis2012optimality}, the computational complexity of~\MinRep{} has remained open since its introduction~\cite{griewank2005analysis,mosenkis2012optimality}.
Like~\MinCost{}, resolving this question has recently been highlighted as an important open problem~\cite{aksoy2023seven}.

\paragraph*{Our Results.}
We show that both \MinCost{} and \MinRep{} are \NP-complete, resolving the key complexity questions which have stood open since 1993~\cite{griewank1993some} and 2005~\cite{griewank2005analysis}, respectively.
Furthermore, we prove that unless the exponential time hypothesis (ETH)\footnote{The ETH is a popular complexity assumption that states that \textsc{3-Sat} cannot be solved in subexponential time. See \Cref{sec:prelim} for more details.} fails, neither problem
admits a subexponential algorithm, i.e., an algorithm with running time
$2^{o(n+m)}$.
We complement our lower bounds by providing~$O^*(2^n)$-time algorithms for both problems.

\section{Preliminaries and Basic Observations}
\label{sec:prelim}
In this section, we define the notation we use throughout the paper, introduce relevant concepts from the existing literature, and show two useful basic propositions.

\paragraph*{Notation.}
For a positive integer~$n$, we use~$[n]$ to denote the set~$\{1,2,\ldots,n\}$.
We use standard graph terminology.
In particular, a graph~$G=(V,E)$ or~$D=(V,A)$ is a pair where~$V$ denotes the set of vertices and~$E$ and~$A$ denote the set of (undirected) edges or (directed) arcs, respectively.
We use~$n$ to indicate the number of vertices in a graph and~$m$ to indicate the number of edges or arcs.
For an (undirected) edge between two vertices~$u$ and~$v$ we write~$\{u, v\}$, and for an arc (a directed edge) from~$u$ to~$v$ we write~$(u, v)$.
Given a vertex $v \in V$, we denote by~$\innbrs{D}{v}$ and~$\outnbrs{D}{v}$ the open in- and out-neighborhoods of $v$, respectively.
The \emph{Markowitz degree} is defined to be~$\mu_D(v) =\deg_D^-(v)\cdot\deg_D^+(v) = |\innbrs{D}{v}|\cdot|\outnbrs{D}{v}|$.
If the graph is clear from context, we omit the subscript in the above notation.
We say that two vertices $u$ and $v$ are \emph{false twins} in a directed graph~$D$ if $\innbrs{D}{u} = \innbrs{D}{v}$ and~$\outnbrs{D}{u} = \outnbrs{D}{v}$. 
Given sequences~$\sigma_1=(a_1,a_2,\ldots,a_i)$ and~$\sigma_2=(b_1,b_2,\ldots,b_j)$, we write~$(\sigma_1,\sigma_2)$ for the combined sequence~$(a_1,a_2,\ldots,a_i,b_1,b_2,\ldots,b_j)$, and we generalize this notation to more than two sequences.

\paragraph*{Reductions and the ETH.}
We assume the reader to be familiar with basic concepts in complexity theory like big-O (Bachmann--Landau) notation, \NP-completeness, and polynomial-time many-one reductions (also known as Karp reductions).
We refer to the standard textbook by Garey and Johnson~\cite{GJ79} for an introduction.
We use~$O^*$ to hide factors that are polynomial in the input size and call a polynomial-time many-one reduction a linear reduction when the size of the constructed instance~$I'$ is linear in the size of the original instance~$I$, that is, $|I'| \in O(|I|)$. 
The \emph{exponential time hypothesis~(ETH)}~\cite{IP01} states that there is some~$\varepsilon > 0$ such that each algorithm solving \textsc{3-Sat} takes at least~$2^{\varepsilon n + o(n)}$~time, where~$n$ is the number of variables in the input instance.
Assuming the ETH, \textsc{3-Sat} and many other problems cannot be solved in subexponential ($2^{o(n+m)}$) time~\cite{IPZ01}.
It is known that if there is a linear reduction from a problem~$A$ to a problem~$B$ and~$A$ cannot be solved in subexponential time, then also~$B$ cannot be solved in subexponential time~\cite{IPZ01}.

\paragraph*{Fundamental Observations.}
We next show two useful observations.
The first one states that the order in which vertices are eliminated does not affect the resulting graph (note that it may still affect the cost of the elimination sequence).
This is a folklore result, but to our knowledge no proof is known. Our argument can be seen as an adaptation of one used by Rose and Tarjan to prove a closely related result~\cite{rose1978algorithmic}.

\begin{proposition}
    \label{lem:order}
    Let~$D=(V,A)$ be a DAG, let~$X \subseteq V$ be a set of vertices, and let~$\sigma_1$ and~$\sigma_2$ be two permutations of the vertices in~$X$. Then, $D_{\sigma_1} = D_{\sigma_2}$.
\end{proposition}

\begin{proof}
	We first show for any DAG~$D=(V,A)$ and any three vertices~${u,v,w \in V}$ that there is a directed path from~$u$ to~$v$ in~$D$ if and only if there is a directed path from~$u$ to~$v$ in~$\gElimVertex{D}{w}$.
	To this end, first assume that there is a directed path~$P$ from~$u$ to~$v$ in~$D$.
	If~$P$ does not contain~$w$, then~$P$ is also a directed path in~$\gElimVertex{D}{w}$.
	Otherwise, let~$x,y$ be the vertices before and after~$w$ in~$P$, respectively.
	Since the elimination of~$w$ adds an arc from~$x$ to~$y$, there is also a directed path from~$u$ to~$v$ in this case.
	Now assume that there is a directed path in~$\gElimVertex{D}{w}$.
	We assume without loss of generality that~$P$ is a shortest path in~$D$.
	There are again two cases.
	Either~$P$ is also a directed path in~$D$ or there is at least one arc~$(x,y)$ in~$P$ that is not present in~$D$.
	In the first case, $P$ shows that there is a directed path from~$u$ to~$v$ in~$D$ by definition.
	In the second case, we consider the first arc~$(x,y)$ in~$P$ that is not contained in~$D$.
	Note that by construction~$(x,y)$ is only added to~$\gElimVertex{D}{w}$ if~$x$ is an in-neighbor of~$w$ and~$y$ is an out-neighbor of~$w$ in~$D$.
	Moreover, since~$P$ is a shortest path, there is no other vertex~$y'$ in~$P$ that is also a out-neighbor of~$w$ as otherwise the arc~$(x,y')$ exists in~$\gElimVertex{D}{w}$ contradicting that~$P$ is a shortest path.
	We can now replace the arc~$(x,y)$ in~$P$ by a subpath~$(x,w,y)$ to get a directed path from~$u$ to~$v$ in~$D$.
	
	Let~$\sigma = (w_1,w_2,\ldots,w_k)$ be a sequence.
  By induction on $k$ (using the above argument), it holds for any two vertices~$u,v \in V \setminus \{w_1,w_2,\ldots,w_k\}$ that there is a directed path from~$u$ to~$v$ in~$D$ if and only if there one in~$(((\gElimVertex{D}{w_1})_{/w_2})\ldots)_{/w_k} = D_{\sigma}$.

	Now assume towards a contradiction that~$D_{\sigma_1} \neq D_{\sigma_2}$.
	Note that both contain the same set~$V \setminus X$ of vertices.
	We assume without loss of generality that there is an arc~$(u,v)$ that exists in~$D_{\sigma_1}$ but not in~$D_{\sigma_2}$.
	By the above argument, since~$(u,v)$ appears in~$D_{\sigma_1}$ there is a directed path from~$u$ to~$v$ in~$D$.
	However, since~$(u,v)$ does not appear in~$D_{\sigma_2}$, there is no directed path from~$u$ to~$v$ in~$D$, a contradiction.
	This concludes the proof.
\end{proof}

Let~$\sigma$ be a sequence of internal vertices, and let~$X$ be the set of vertices appearing in~$\sigma$. In the rest of this paper, we may use~$\gElimSet{D}{X}$ to denote the graph~$\gElimSequence{D}{\sigma}$. By~\Cref{lem:order}, this notation is well-defined.

To conclude this section, we show that false twins can be handled uniformly in \MinCost. Let $T$ be a set of false twins, i.e., $u$ and $v$ are false twins for every $u, v \in T$. Then we may assume, without loss of generality, that eliminating any $u \in T$ also entails eliminating the rest of the vertices in $T$ immediately afterward.

\begin{proposition}\label{lem:twins}
  Let~$D=(S \uplus I, A)$ be a DAG and let $T \subseteq I$ be a set of false twins. Then, there exists an optimal elimination sequence (for \MinCost) that eliminates the vertices of $T$ consecutively.
\end{proposition}

\begin{proof}
Let $T \subseteq I$ be a set of false twins in~$D$. 
We first prove the result when~$|T| = 2$. Let $T = \{u, v\}$, and let $\sigma$ be an optimal solution. We may assume that~$u$ and $v$ are eliminated non-consecutively in $\sigma$, as otherwise the proof is complete.
We further assume without loss of generality that~$u$ is eliminated before~$v$ in~$\sigma$.
Let~$\sigma_1$ be the subsequence of~$\sigma$ before~$u$, $\sigma_2$ be the subsequence between~$u$ and~$v$ and~$\sigma_3$ be the subsequence after~$v$.
Let~$X_1,X_2$, and~$X_3$ be the sets of vertices appearing in~$\sigma_1,\sigma_2$, and~$\sigma_3$, respectively.
Note that~$X_1$ and~$X_3$ might be empty, but~$X_2$ contains at least one vertex.
Let $\sigma'=(\sigma_1,\sigma_2,u,v,\sigma_3)$ and~$\sigma''=(\sigma_1,u,v,\sigma_2,\sigma_3)$.
Let $c, c'$, and~$c''$ be the costs of $\sigma, \sigma'$, and $\sigma''$, respectively. 
Since $\sigma$ is optimal, it holds that~$c \leq c'$ and $c \leq c''$.
Now, we claim that both~$\sigma'$ and $\sigma''$ are optimal (it suffices to show that either $\sigma'$ or $\sigma''$ is optimal, but it will be useful later to prove that both are).
Assume otherwise, so at least one of the inequalities $c \leq c'$ or $c \leq c''$ is strict.
This implies~$2c < c' + c''$.
We will show that this inequality leads to a contradiction.

It holds that the cost of~$\sigma_1$ in $D$ and the cost of~$\sigma_3$ in~$D_{X_1 \cup X_2 \cup \{u,v\}}$ are accounted for identically in the total costs of $\sigma, \sigma'$, and~$\sigma''$. Thus, any difference in the values of~$c,c',$ and~$c''$ is attributable entirely to differing costs of eliminating $u, v$ and the vertices in $X_2$.
Let~$d_1,d_2,$ and~$d_3$ be the costs of~$\sigma_2$ in~$D_{X_1},D_{X_1 \cup \{u\}}$, and~$D_{X_1 \cup \{u,v\}}$, respectively. Note that these terms are well-defined by \cref{lem:order}.
This implies
\begin{align}
\begin{split}
2\big(\mu_{D_{X_1}}(u) + d_2 + \mu_{D_{X_1 \cup \{u\} \cup X_2}}(v)\big) <&\\ \big(d_1 + \mu_{D_{X_1 \cup X_2}}(u) + \mu_{D_{X_1 \cup X_2 \cup \{u\}}}&(v)\big) + \big( \mu_{D_{X_1}}(u) + \mu_{D_{X_1 \cup \{u\}}}(v)+d_3\big). \label{eq:cost}
\end{split}
\end{align}

Moreover, since~$u$ and~$v$ are false twins and the elimination of~$u$ does not change the cost of eliminating~$v$, it holds that~$\mu_{D_{X_1 \cup X_2}}(u) = \mu_{D_{X_1 \cup X_2 \cup \{u\}}}(v)$ and~${\mu_{D_{X_1}}(u) = \mu_{D_{X_1 \cup \{u\}}}(v)}$.
Substituting this into Inequality~(\ref{eq:cost}) yields~$2d_2 < d_1 + d_3$.
Next, let~${\sigma_2 = (w_1,w_2,\ldots,w_k)}$ and let~${W_i = X_1 \cup \{w_1,w_2,\ldots,w_{i-1}\}}$ for each~$i \in [k]$.
Notice that it holds that~${d_1 = \sum_{i \in [k]} \mu_{D_{W_i}}(w_i)}$, ${d_2 = \sum_{i \in [k]} \mu_{D_{W_i \cup \{u\}}}(w_i)}$, and~$d_3 =  \sum_{i \in [k]} \mu_{D_{W_i \cup \{u,v\}}}(w_i)$.
To conclude the proof, we will show that for each~$w_i \in X_2$,
\[2 \mu_{D_{W_i \cup \{u\}}}(w_i) \geq \mu_{D_{W_i}}(w_i) +  \mu_{D_{W_i \cup \{u,v\}}}(w_i).\]
Note that this implies~$2d_2 \geq d_1 + d_3$, yielding the desired contradiction.
To show the above claim, we consider three cases: (i) $w_i$ is an out-neighbor of~$u$ in~$D_{W_i}$, (ii) $w_i$ is an in-neighbor of~$u$ in~$D_{W_i}$, and (iii) $w_i$ is neither an in- nor an out-neighbor of~$u$ in~$D_{W_i}$.
Note that since~$D_{W_i}$ is a DAG, $w_i$ cannot be both an in-neighbor and an out-neighbor of~$u$.
Moreover, since~$u$ and~$v$ are false twins, $w_i$ is an in-/out-neighbor of~$u$ if and only if it is an in-/out-neighbor of~$v$.
In the first case, note that
\begin{itemize}
	\item $|\outnbrs{D_{W_i}}{w_i}| = |\outnbrs{D_{W_i \cup \{u\}}}{w_i}| = |\outnbrs{D_{W_i \cup \{u,v\}}}{w_i}|$, 
	\item $|\innbrs{D_{W_i}}{w_i}| \leq |\innbrs{D_{W_i \cup \{u,v\}}}{w_i}|+2$, and
	\item $|\innbrs{D_{W_i \cup \{u\}}}{w_i}| = |\innbrs{D_{W_i \cup \{u,v\}}}{w_i}| + 1$.
\end{itemize}
The first holds as the out-degree of~$w_i$ does not change if~$u$ and/or~$v$ are eliminated.
To see the second, note that eliminating~$u$ and~$v$ can reduce the in-degree of $w_i$ by at most two. Finally, if~$u$ is already eliminated, then eliminating~$v$ does not add any new in-neighbors of~$w_i$ since~$u$ and~$v$ are false twins (and this property remains true even if other vertices are eliminated).
Thus, we get
\begin{align*}
	2 \mu_{D_{W_i \cup \{u\}}}(w_i) &= 2\big(  |\innbrs{D_{W_i \cup \{u\}}}{w_i}| \cdot |\outnbrs{D_{W_i \cup \{u\}}}{w_i}|\big)\\
	&= 2\big(  (|\innbrs{D_{W_i \cup \{u,v\}}}{w_i}| + 1) \cdot |\outnbrs{D_{W_i \cup \{u\}}}{w_i}|\big)\\
	&= (2|\innbrs{D_{W_i \cup \{u,v\}}}{w_i}| + 2) \cdot |\outnbrs{D_{W_i \cup \{u\}}}{w_i}|\\
	&\geq (|\innbrs{D_{W_i}}{w_i}| + |\innbrs{D_{W_i \cup \{u,v\}}}{w_i}|) \cdot |\outnbrs{D_{W_i \cup \{u\}}}{w_i}|\\
	&= \mu_{D_{W_i}}(w_i) +  \mu_{D_{W_i \cup \{u,v\}}}(w_i).
\end{align*}

The second case is analogous with the roles of in- and out-neighbors swapped, that is, 
\begin{itemize}
	\item $|\innbrs{D_{W_i}}{w_i}| = |\innbrs{D_{W_i \cup \{u\}}}{w_i}| = |\innbrs{D_{W_i \cup \{u,v\}}}{w_i}|$, 
	\item $|\outnbrs{D_{W_i}}{w_i}| \leq |\outnbrs{D_{W_i \cup \{u,v\}}}{w_i}|+2$, and
	\item $|\outnbrs{D_{W_i \cup \{u\}}}{w_i}| = |\outnbrs{D_{W_i \cup \{u,v\}}}{w_i}| + 1$.
\end{itemize}
This yields
\begin{align*}
	2 \mu_{D_{W_i \cup \{u\}}}(w_i) &= 2\big(  |\innbrs{D_{W_i \cup \{u\}}}{w_i}| \cdot |\outnbrs{D_{W_i \cup \{u\}}}{w_i}|\big)\\
	&= 2\big(|\innbrs{D_{W_i \cup \{u\}}}{w_i}| \cdot (|\outnbrs{D_{W_i \cup \{u,v\}}}{w_i}|+1)\big)\\
	&= |\innbrs{D_{W_i \cup \{u\}}}{w_i}| \cdot (2|\outnbrs{D_{W_i \cup \{u,v\}}}{w_i}|+2)\\
	&\geq |\innbrs{D_{W_i \cup \{u\}}}{w_i}| \cdot  (|\outnbrs{D_{W_i}}{w_i}| + |\outnbrs{D_{W_i \cup \{u,v\}}}{w_i}|)\\
	&= \mu_{D_{W_i}}(w_i) +  \mu_{D_{W_i \cup \{u,v\}}}(w_i).
\end{align*}

Finally, in the third case it holds that
\begin{itemize}
	\item $|\innbrs{D_{W_i}}{w_i}| = |\innbrs{D_{W_i \cup \{u\}}}{w_i}| = |\innbrs{D_{W_i \cup \{u,v\}}}{w_i}|$, and
	\item $|\outnbrs{D_{W_i}}{w_i}| = |\outnbrs{D_{W_i \cup \{u\}}}{w_i}| = |\outnbrs{D_{W_i \cup \{u,v\}}}{w_i}|$.
\end{itemize}

Thus, we get
\[2 \mu_{D_{W_i \cup \{u\}}}(w_i) = 2\cdot|\innbrs{D_{W_i \cup \{u\}}}{w_i}| \cdot |\outnbrs{D_{W_i \cup \{u\}}}{w_i}| = \mu_{D_{W_i}}(w_i) + \mu_{D_{W_i \cup \{u,v\}}}(w_i).\]
Since~$2 \mu_{D_{W_i \cup \{u\}}}(w_i) \geq \mu_{D_{W_i}}(w_i) + \mu_{D_{W_i \cup \{u,v\}}}(w_i)$ holds in all cases and we showed before that this contradicts Inequality~(\ref{eq:cost}), this concludes the proof when $|T| = 2$.

To generalize the result to larger sets~$T$, we use induction on the size~$\ell$ of~$T$.
As shown above, the base case where~$\ell=2$ holds true.
Next, assume the proposition is true for~$\ell-1$.
Let~$T = \{v_1, v_2, \ldots, v_\ell\}$, and let $\sigma$ be an optimal elimination sequence.
Let~$T' = T \setminus \{v_\ell\}$.
By the inductive hypothesis (on~$T'$) and \cref{lem:order}, we may assume that~$\sigma = (\sigma_1, v_1, v_2, \ldots v_{\ell - 1}, \sigma_2, v_\ell, \sigma_3)$ or~${\sigma = (\sigma_1, v_\ell, \sigma_2, v_1, v_2, \ldots, v_{\ell-1}, \sigma_3)}$ for some sequences~$\sigma_1,\sigma_2$, and~$\sigma_3$.
By applying the argument from the case where~$|T|=2$ to~$v_{\ell-1}$ and~$v_{\ell}$ or to~$v_\ell$ and~$v_1$, we observe optimal solutions~$(\sigma_1, v_1, v_2, \ldots, v_{\ell-1},v_{\ell}, \sigma_2, \sigma_3)$ or~${(\sigma_1, \sigma_2, v_\ell, v_1, v_2, \ldots, v_{\ell-1}, \sigma_3)}$.
This completes the proof.
\end{proof}

\section{\MinCost{} is \NP-complete}
In this section, we show that \MinCost{} is \NP-complete.
We reduce from \textsc{Vertex Cover}, which is defined as follows.

\bigskip
\defproblem{\textsc{Vertex Cover}}
{An undirected graph~$G=(V,E)$ and an integer~$k$.}
{Is there a set~$C \subseteq V$ of at most~$k$ vertices such that each edge in~$E$ has at least one endpoint in~$C$?}

\begin{theorem}
    \MinCost{} is \NP-complete. Assuming the ETH, it cannot be solved in~$2^{o(n+m)}$ time.\looseness=-1
\end{theorem}

\begin{proof}
    We first show containment in \NP. Note that a total elimination sequence is a permutation of~$I$ and therefore can be encoded in polynomial space.
    Moreover, given such a sequence, we can compute its cost in polynomial time by computing the vertex eliminations one after another.
    
    To show hardness, we reduce from \textsc{Vertex Cover}.
    It is well known that \textsc{Vertex Cover} is \NP-hard and cannot be solved in~$2^{o(n+m)}$ time unless the ETH fails~\cite{IPZ01,Karp72}.
    We will provide a linear reduction from \textsc{Vertex Cover} to \MinCost{} thereby proving the theorem.
    To this end, let~$(G=(V,E),k)$ be an input instance of \textsc{Vertex Cover}.
    Let~$n = |V|$ and $m = |E|$.
    We create an equivalent instance~$(D=(S \uplus I,A),k')$ of \MinCost{} as follows.
    For each vertex~$v \in V$, we create five vertices~$v_1,v_2,v_3,v_4,v_5$.
    The vertices~$v_1,v_4$ and~$v_5$ are contained in~$S$ for all~$v \in V$.
    Vertices~$v_2$ and~$v_3$ are contained in~$I$.
    Next, we add the set
    \(
      \left\{(v_1,v_2),(v_2,v_3),(v_2,v_4),(v_2,v_5),(v_3,v_4),(v_3,v_5)\right\}
    \)
    of arcs to~$A$ for each~$v\in V$.
    Finally, for each edge~$\{u,v\} \in E$, we add the arcs~$(u_1,v_3),(u_2,v_3), (v_1,u_3)$ and~$(v_2,u_3)$ to~$A$.
    Notice that every arc goes from a lower-indexed vertex to a higher-indexed vertex.
    Hence, the constructed digraph is a DAG.
    To finish the construction, we set
    \(
      k' = 6m + 4n + k.
    \)
    An illustration of the construction is depicted in \Cref{fig:example}.
    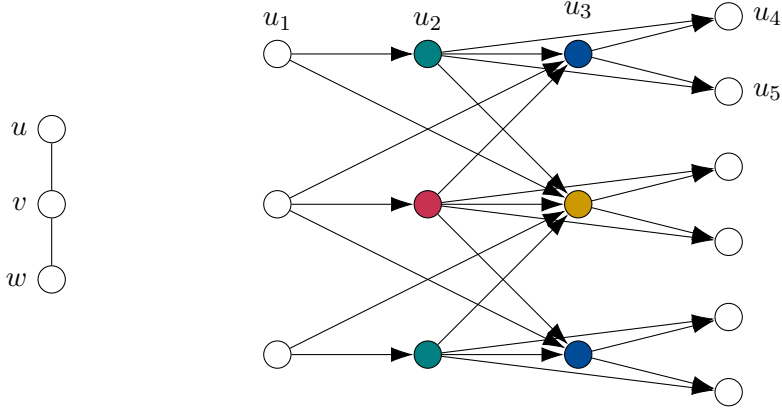
\begin{figure}[t]
        \centering
        \begin{tikzpicture}
            \node[circle,draw,label=left:$u$] at(-3,-3) (u) {};
            \node[circle,draw,label=left:$v$] at(-3,-4) (v) {} edge (u);
            \node[circle,draw,label=left:$w$] at(-3,-5) (w) {} edge (v);

            \foreach \i/\c/\d in {1/cbgreen/cbblue,2/cbred/cbyellow,3/cbgreen/cbblue}{
                \node[circle,draw,label=above:\ifthenelse{\equal{\i}{1}}{$u_1$}{}] (v\i1) at (0,-2*\i) {};
                \node[circle,draw,fill=\c,label=above:\ifthenelse{\equal{\i}{1}}{$u_2$}{}] (v\i2) at (2,-2*\i) {};
                \node[circle,draw,fill=\d,label={[label distance = .15cm]above:\ifthenelse{\equal{\i}{1}}{$u_3$}{}}] (v\i3) at (4,-2*\i) {};
                \node[circle,draw,label=right:\ifthenelse{\equal{\i}{1}}{$u_4$}{}] (v\i4) at (6,-2*\i+0.5) {};
                \node[circle,draw,label=right:\ifthenelse{\equal{\i}{1}}{$u_5$}{}] (v\i5) at (6,-2*\i-0.5) {};
                \draw[->] (v\i1) to (v\i2);
                \draw[->] (v\i2) to (v\i3);
                \draw[->] (v\i2) to (v\i4);
                \draw[->] (v\i2) to (v\i5);
                \draw[->] (v\i3) to (v\i4);
                \draw[->] (v\i3) to (v\i5);
            }
            \foreach \i/\j in {1/2,2/1,2/3,3/2}{ 
                \draw[->] (v\i1) to (v\j3);
            }
            \foreach \i/\j in {1/2,2/3}{
                \draw[->] (v\i2) to (v\j3);
                \draw[->] (v\j2) to (v\i3);
            }
        \end{tikzpicture}
        \caption{The input graph is shown on the left and the constructed graph is shown on the right. An optimal solution (corresponding to the vertex cover that only contains the middle vertex) first eliminates the red vertex, then all blue vertices, then all green vertices, and finally the yellow vertex.}
        \label{fig:example}
    \end{figure}

    Since the reduction can clearly be computed in polynomial time, it only remains to show correctness.  We proceed to show that $(G, k)$ is a yes-instance of \textsc{Vertex Cover} if and only if the constructed instance~$(D,k')$ is a yes-instance of \MinCost.

    First, assume that~$G$ contains a vertex cover~$C$ of size at most~$k$.
    We show that eliminating all vertices~$v_2$ with~$v \in C$ first, then~$u_3$ for all~$u \in V \setminus C$, followed by~$u_2$ for all~$u \in V \setminus C$, and finally all vertices~$v_3$ for each~$v \in C$ results in a total cost of~$k'$.
    To see this, note that the cost of eliminating~$v_2$ for any vertex~$v \in C$ is~$\deg(v)+3$ as~$v_2$ has a single in-neighbor~$v_1$, three out-neighbors~$v_3,v_4$, and~$v_5$, and~$\deg(v)$ out-neighbors~$u_3$ (one for each~$u \in N(v)$).
    The cost of eliminating~$u_3$ for any~$u \notin C$ afterwards is~$2\deg(u) + 2$ as by construction~$u_3$ has in-neighbors~$\{w_1 \mid w \in N(u)\} \cup \{u_2\}$, two out-neighbors~$u_4,u_5$ and no in-neighbors in $\{w_2 \mid w \in N(u)\}$, as each $w \in N(u)$ is by definition in~$C$ and hence the corresponding $w_2$ has been eliminated before.
    The cost of eliminating~$u_2$ for~$u \in V \setminus C$ is afterwards~$\deg(u)+2$ as~$u_2$ has the single in-neighbor~$u_1$, two out-neighbors~$u_4$ and~$u_5$ and~$\deg(u)$ out-neighbors~$\{v_3 \mid v \in C\}$ (note that~$u$ cannot have neighbors in~$V \setminus C$ since~$C$ is a vertex cover and~$u \notin C$).
    Finally, the cost of eliminating~$v_3$ for any~$v \in C$ is~$2\deg(v)+2$ since~$v_3$ has~$\deg(v)+1$ in-neighbors $\{w_1 \mid w \in N[v]\}$ and two out-neighbors~$v_4$ and~$v_5$.
    Summing these costs over all vertices and applying the handshake lemma\footnote{The handshake lemma states that the sum of vertex degrees is twice the number of edges~\cite{diestel2012graph}.} gives a total cost of
    \begin{align*}
      \sum_{v \in C} (\deg(v)+3) &+ \sum_{u \in V \setminus C} (2\deg(u)+2)
      + \sum_{u \in V \setminus C} (\deg(u)+2) + \sum_{v \in C} (2\deg(v)+2)\\
    &= \sum_{v \in C} (3\deg(v) + 5) + \sum_{u \in V\setminus C} (3\deg(u) +4) \\
    &= \sum_{v \in V} (3\deg(v)+4) + |C|\\
    &\leq 6m+4n+k=k'.
    \end{align*}
    This shows that the constructed instance~$(D,k')$ is a yes-instance of \MinCost.

    In the other direction, assume that there is an ordering~$\sigma$ of the vertices in~$I$ resulting in a total cost of at most~$k'$.
    Let~$J \subseteq V$ be the set of vertices such that for each~$v \in J$ it holds that~$v_2$ is eliminated before~$v_3$ or~$v_3$ is eliminated before~$u_2$ for any~$u \in N(v)$ by~$\sigma$.
    We will show that~$J$ is a vertex cover of size at most~$k$ in~$G$.
    To this end, we first provide a lower bound for the cost of eliminating any vertex, regardless of which vertices have been eliminated previously.
    Note that~$v_3$ for any vertex~$v \in V$ has two out-neighbors~$v_4$ and~$v_5$ in~$S$, $\deg(v)$ in-neighbors $\{w_1 \mid w \in N(v)\}$ in~$S$, and one additional in-neighbor which is either~$v_2$ if~$v_2$ was not eliminated before or~$v_1$ if~$v_2$ was eliminated before.
    Hence, the cost for eliminating~$v_3$ is at least~$2\deg(v)+2$.
    Moreover, the cost of eliminating~$v_2$ for any vertex~$v \in V$ is at least~$\deg(v)+2$ as~$v_2$ has the in-neighbor~$v_1 \in S$, two out-neighbors~$v_4,v_5 \in S$, and for each~$w \in N(v)$ at least one additional out-neighbor ($w_3$ if~$w_3$ was not eliminated before or~$w_4$ and~$w_5$ if~$w_3$ was eliminated before).
    Summing these costs over all vertices (and again applying the handshake lemma) gives a lower bound of~$6m+4n = k' - k$.

    The next step is to prove that~$J$ contains at most~$k$ vertices.
    To this end, note that for each vertex~$v \in J$, the cost increases by at least one over the analyzed lower bound.
    If~$v_3$ is eliminated after~$v_2$ for some~$v \in J$, then the cost of eliminating~$v_2$ increases by one as~$v_2$ has the additional out-neighbor~$v_3$.
    If~$v_3$ is eliminated before~$u_2$ for some~$v \in J$ and~$u \in N(v)$, then the cost of eliminating~$u_2$ increases by one as the out-neighbor~$v_3$ is replaced by the two out-neighbors~$v_4$ and~$v_5$.
    This immediately implies that~$|J| \leq k$.

    Finally, we show that~$J$ is a vertex cover.
    Assume towards a contradiction that this is not the case.
    Then, there is some edge~$\{u,v\} \in E$ with~$u \notin J$ and~$v \notin J$.
    By definition of~$J$, it holds that~$u_3$ is eliminated before~$u_2$, $u_3$ is eliminated after~$v_2$, $v_3$ is eliminated before~$v_2$, and~$v_3$ is eliminated after~$u_2$ by~$\sigma$.
    Note that this implies that~$\sigma$ eliminates~$u_3$ before~$u_2$, before~$v_3$, before~$v_2$, before~$u_3$, a contradiction.
    Thus, $J$ is a vertex cover of size at most~$k$ and the initial instance of \textsc{Vertex Cover} is therefore a yes-instance.
    This concludes the proof.
\end{proof}

\section{\MinRep{} is \cclass{NP}-complete}\label{sec:min-edge-count-hard}
In this section we show that \MinRep{} is \NP-complete and, assuming the ETH, it cannot be solved in subexponential time.
To this end, we reduce from \textsc{Independent Set}, which is defined as follows.

\bigskip
\defproblem{\textsc{Independent Set}}
{An undirected graph~$G=(V,E)$ and an integer~$k$.}
{Is there a set~$X \subseteq V$ of at least~$k$ vertices such that no edge in~$E$ has both endpoints in~$X$?}

\begin{theorem}
  \MinRep{} is \NP-complete. Assuming the ETH, it cannot be solved in~$2^{o(n+m)}$ time.
\end{theorem}

\begin{proof}
	We again start by showing containment in \NP.
	We can encode a (not necessarily total) elimination sequence in polynomial space.
	Moreover, given such a sequence, we can compute the resulting DAG in polynomial time and verify that it contains at most~$k$ edges.

  To show hardness, we reduce from \textsc{Independent Set} in 2-degenerate subcubic graphs of minimum girth five, that is, graphs in which each vertex has between two and three neighbors and each cycle has length at least five.
  This problem is known to be \NP-hard and cannot be solved in~$2^{o(n+m)}$ time assuming the ETH~\cite{Kom18}.
  We will provide a linear reduction from that problem to \MinRep{} to show the theorem.
  To this end, let $(G = (V=\{v_1,v_2,\ldots,v_n\}, E), k)$ be an instance of \textsc{Independent Set} where each vertex has between two and three neighbors and no cycles of length three or four are present in~$G$.
  We will construct an instance~${(D=(S \uplus I, A), k')}$ of \MinRep.
  We begin by imposing an arbitrary total order~$\pi$ on the vertex set~$V$.
  We partition the vertices into four types based on their degree and the order~$\pi$ as follows.
  Vertices of type~1 have degree~3 and either all neighbors come earlier with respect to~$\pi$ or all neighbors come later.
  Vertices of type~2 have degree~3 and at least one earlier and at least one later neighbor with respect to~$\pi$.
  Vertices of type~3 have degree~2 and either both neighbors come earlier or both neighbors come later with respect to~$\pi$.
  Finally, vertices of type~4 have degree~2 and one of the neighbors comes earlier with respect to~$\pi$ while the other neighbor comes later.

  We now describe the construction of $D$, depicted in~\Cref{fig:MEC-reduction}.
  We begin by creating a set $T \subseteq S$ of $4$ sink vertices.
  Next, for each~$v_i \in V$, we create a vertex $u_i \in I$ as well as two vertex sets $I_i \subseteq S$ and $O_i \subseteq S$ both of size $4$.
  We add arcs from each vertex in~$I_i$ to every vertex in~$\{u_i\} \cup T$ to~$A$.
  We also add arcs from $u_i$ to each vertex in~$O_i \cup T$ to~$A$.
  Finally, we add arcs from~$I_i$ to~$O_i$ based on the type of~$v_i$ as follows.
  Note that since~$|I_i| = |O_i| = 4$, there are up to~$16$ possible arcs from vertices in~$I_i$ to vertices in~$O_i$.
  For type-1 vertices, we add~14 of the possible arcs (it does not matter which arcs we add).
  For vertices of type~2,3, and~4, we add~$16,11,$ and~$12$ arcs to~$A$, respectively.
  After completing this procedure for every vertex in $V$, we add an arc~$(u_i,u_j)$ for any edge~$\{v_i,v_j\} \in E$ where~$v_i$ comes before~$v_j$ in the order~$\pi$.
  This concludes the construction of the graph~$D$.
  Let~$n'$ and~$m'$ be the number of vertices and arcs in~$D$.
  To conclude the construction, we set~$k'=m'-k$.
  Observe that $D$ is a DAG and the number of vertices and arcs is linear in~$n+m$.
  \begin{figure}[t]
      \centering
    \begin{tikzpicture}
        \node[circle,inner sep=3pt,draw,label=$v_1$] (v1) at(-5,3) {};
        \node[circle,inner sep=3pt,draw,label=below:$v_2$] (v2) at(-6,2) {} edge(v1);
        \node[circle,inner sep=3pt,draw,label=below:$v_3$] (v3) at(-4,2) {} edge(v1) edge(v2);

          \node[circle,inner sep=5pt,draw,label=below:$T$,fill=gray] (t) at(0,0) {};
          \node[circle,inner sep=5pt,draw,label=below:$I_1$,fill=gray] (i1) at(-2,-1) {};
          \node[circle,inner sep=3pt,draw,label=$u_1$] (u1) at(-2,1) {};
          \node[circle,inner sep=5pt,draw,label=$O_1$,fill=gray] (o1) at(-4,0) {};
          \node[circle,inner sep=5pt,draw,label=left:$I_2$,fill=gray] (i2) at(-1,2) {};
          \node[circle,inner sep=3pt,draw,label=right:$u_2$] (u2) at(1,2) {};
          \node[circle,inner sep=5pt,draw,label=$O_2$,fill=gray] (o2) at(0,4) {};
          \node[circle,inner sep=5pt,draw,label=below:$I_3$,fill=gray] (i3) at(2,-1) {};
          \node[circle,inner sep=3pt,draw,label=$u_3$] (u3) at(2,1) {};
          \node[circle,inner sep=5pt,draw,label=$O_3$,fill=gray] (o3) at(4,0) {};

          \draw[->] (u1) to (u3);
          \draw[->] (u1) to (u2);
          \draw[->] (u2) to (u3);
          \foreach \i in {1,2,3}{
              \draw[->,line width = .75mm] (i\i) to (t);
              \draw[->,line width = .75mm] (i\i) to (u\i);
              \draw[->,line width = .75mm] (u\i) to (t);
              \draw[->,line width = .75mm] (u\i) to (o\i);
          }
          \draw[->,line width = .75mm] (i1) to node[midway, below=3pt]{\emph{11}} (o1);
          \draw[->,line width = .75mm] (i2) to node[midway, left=3pt]{\emph{12}} (o2);
          \draw[->,line width = .75mm] (i3) to node[midway, below=3pt]{\emph{11}} (o3);
      \end{tikzpicture}
      \caption{\label{fig:MEC-reduction}An example instance of independent set on the left and the constructed instance on the right. We assume the order~$\pi$ to be~$v_1$ first, then~$v_2$ and~$v_3$ last. Each large node~$I_i,O_i$ and~$T$ (shaded gray) represents an independent set of size~$4$. A bold arc between two nodes represents all possible arcs between the respective vertex sets (in one direction) unless a number is shown next to the arc. In this case, the number represents the number of arcs between the two sets of vertices.}
  \end{figure}
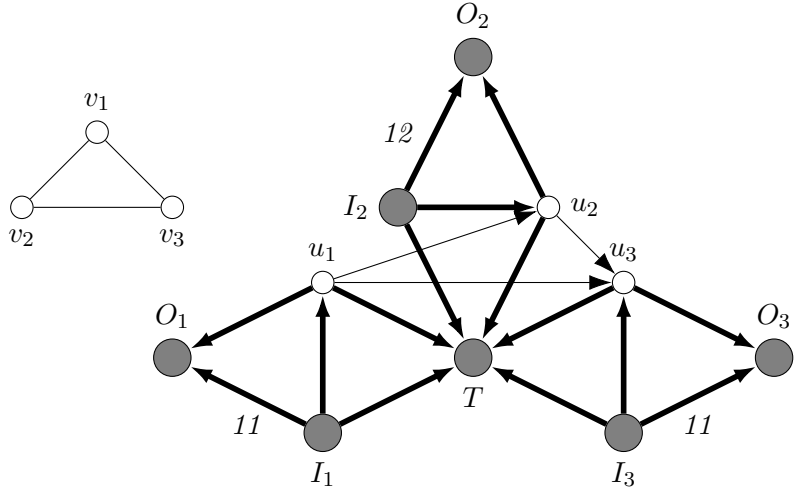

  Since the reduction can be computed in polynomial time and the constructed instance has linear size in the input instance, it only remains to show correctness.
  To this end, first assume that there exists an independent set $X \subseteq V$ of size $k$ in $G$.
  We eliminate the vertices~$X' = \{u_i \mid v_i \in X\}$.
  Note that \cref{lem:order} ensures that the order in which we eliminate the vertices does not matter.
  We will show that the resulting graph contains at most~$k'$ arcs.

    First, note that new arcs between vertices~$u_i$ and~$u_j$ might be created while eliminating vertices in~$X'$.
      However, since~$X$ is an independent set, such arcs are only created between vertices where~$N(v_i) \cap N(v_j) \neq \emptyset$.
      Since~$G$ has girth at least~5, if the elimination of a vertex in~$X'$ could create the arc~$(u_i,u_j)$, then this arc was not there initially and was not added by the elimination of a different vertex in~$X'$.
      We next show that the elimination of each vertex in~$X'$ reduces the number of arcs in the graph by exactly one.
      Let~$u_i \in X'$ be an arbitrary vertex.
      If~$v_i$ has degree three, then~$u_i$ has exactly~$15$ incident arcs, where 12 are to or from~$I_i,O_o$ and~$T$ and three are to vertices~$u_{j_1},u_{j_2},u_{j_3}$.
      The number of new arcs created in this case is~$14$ as shown next.
      For each~$\ell \in [3]$, four new arcs are created from vertices in~$I_i$ to~$u_{j_\ell}$ if~$v_i$ comes before~$v_{j_\ell}$ with respect to~$\pi$ and four new arcs from~$u_{j_\ell}$ to vertices in~$O_i$ are created otherwise.
      Hence, in any case, 12 new arcs are created.
      If~$v_i$ is of type~1, then 2 arcs are created from vertices in~$I_i$ to vertices in~$O_i$.
    If~$v_i$ is of type~2, then two additional arcs are created between the vertices~$u_{j_1},u_{j_2}$, and~$u_{j_3}$.
    Hence, if~$v_i$ has degree 3, then the elimination of vertex~$u_i$ removes~$15$ arcs and creates 14 new ones, that is, the number of arcs decreases by one.
      If~$v_i$ has degree 2, then~$u_i$ is incident to exactly~$14$ arcs (12 to and from vertices in~$I_i,O_i$, and~$T$ and two additional arcs to or from vertices~$u_{j_1}$ and~$u_{j_2}$).
      The number of new arcs created in this case is~$13$ as shown next.
      For each~$\ell \in [2]$, four new arcs are created from vertices in~$I_i$ to~$u_{j_\ell}$ if~$v_i$ comes before~$v_{j_\ell}$ with respect to~$\pi$ and four new arcs from~$u_{j_\ell}$ to vertices in~$O_i$ are created otherwise.
      In any case, 8 new arcs are created this way.
      If~$v_i$ is of type~3, then five arcs are created from vertices in~$I_i$ to vertices in~$O_i$.
      If~$v_i$ is of type~4, then four additional arcs are created from vertices in~$I_i$ to vertices in~$O_i$ and one additional arc is added between~$u_{j_1}$ and~$u_{j_2}$.
    Hence, if~$v_i$ has degree 2, then the elimination of vertex~$u_i$ removes~$14$ arcs and creates 13 new ones, that is, the number of arcs also decreases by one in this case.
      Since~$k' = m'-k$ and each of the~$k$ removals decreases the number of arcs by one, the resulting graph contains at most~$k'$ arcs, showing that the constructed instance is a yes-instance.

  For the other direction, assume that $X'$ is a solution to $(D, k')$, that is, $D_{X'}$ has at most $k'$ arcs. Further, assume that $X'$ is \emph{minimal} in the sense that, for each~$u \in X'$, $D_{X'\setminus\{u\}}$ contains more arcs than~$D_{X'}$.
  Note that this notion is well-defined due to \cref{lem:order}.
  Moreover, given any solution, a minimal one can be computed in polynomial time.
  Let $X = \{v_i \mid u_i \in X'\}$. We will show that~$X$ induces an independent set in~$G$ and that~$|X| \geq k$, that is, $X$ is an independent set of size at least~$k$ in~$G$ and the original instance is therefore a yes-instance.

  Assume toward a contradiction that $X$ is not an independent set in $G$, that is, there exist vertices~$u_i, u_j \in X'$ such that $\{v_i,v_j\} \in E$.
  We claim that~$X'$ is not minimal in this case.
  To prove this claim, note that eliminating a vertex~$u_i$ does not decrease the in-degree or out-degree of any vertex~$u_j$ (at any stage during a elimination sequence) and if~$\{v_i,v_j\} \in E$, then one of the degrees of~$u_j$ increases.
  If~$u_i$ is neither an in-neighbor nor an out-neighbor of~$u_j$, then eliminating~$u_i$ does not change either degree of~$u_j$.
  If~$u_i$ is an in-neighbor, then the out-degree of~$u_j$ remains unchanged and the in-degree increases as the vertices in~$I_i$ become in-neighbors of~$u_j$ (and they cannot be in-neighbors of~$u_j$ while~$u_i$ is not eliminated).
  If~$u_i$ is an out-neighbor of~$u_j$, then the in-degree of~$u_j$ remains unchanged and the out-degree increases as the vertices in~$O_i$ become new out-neighbors of~$u_j$.
  Let~$d$ be the number of vertices~$w$ such that~$(w,u_j)$ or~$(u_j,w)$ is an arc in~$D_{X' \setminus \{u_j\}}$ and~$w \notin I_j \cup O_j \cup T$.
  Note that~$d > \deg(v_j)$ since~$u_i \in X' \setminus \{u_j\}$ and~$\{v_i,v_j\} \in E$.
  Eliminating~$u_j$ in~$D_{X' \setminus \{u_j\}}$ removes~$d + 12$ arcs.
  If~$v_j$ has degree~3, then eliminating~$u_j$ in~$D_{X' \setminus \{u_j\}}$ creates at least~$4d = d + 3d \geq d + 12$ arcs since (i)~$d > \deg(v_j) = 3$ implies~$d \geq 4$ and (ii) eliminating~$u_j$ creates four arcs between vertices in~$I_j \cup O_j$ and each other (in- or out-)neighbor of~$u_j$ except for vertices in~$T$.
  If~$v_j$ has degree 2, then eliminating~$u_j$ creates at least four arcs from vertices in~$I_j$ to vertices in~$O_j$ plus at least~$4d$ arcs, that is, at least~$4d + 4 = d+3d+4 \geq d+13 > d+12$ arcs since~$d > \deg(v_j) = 2$ implies~$d \geq 3$.
  Hence, in any case the number of newly created arcs is at least as large as the number of removed arcs.
  That is, the number of arcs does not decrease, showing that~$X'$ is not a minimal solution.

  It only remains to show that~$|X| \geq k$.
  As analyzed in the forward direction, the elimination of any vertex~$u_i$ reduces the number of arcs by exactly one if no vertex~$u_j$ with~$\{v_i,v_j\} \in E$ was eliminated before.
  Since~$k' = m-k$, this shows that~$|X| \geq k$, concluding the proof.
\end{proof}

\section{Algorithms}
\label{sec:algorithm}

In this section, we give two simple algorithms that show that \MinCost{} and \MinRep{} can be solved in~$O^*(2^n)$ time.
We begin with \MinRep.

\begin{proposition}
	\MinRep{} can be solved in~$O(2^nn^3)$ time and with polynomial space.
\end{proposition}

\begin{proof}
	By \cref{lem:order}, the order in which vertices in an optimal solution are eliminated is irrelevant. Hence, we can simply test for each subset~$X$ of vertices, how many arcs remain if the vertices in~$X$ are eliminated.
	Since there are~$2^n$ possible subsets and each of the at most~$n$ eliminations for each subset can be computed in~$O(n^2)$ time, all subsets can be tested in~$O(2^nn^3)$ time.
\end{proof}

We continue with \MinCost, where we use an algorithmic framework due to Bodlaender et al.~\cite{bodlaender2012note}.

\begin{proposition}
	\MinCost{} can be solved in~$O(2^nn^4)$ time. It can also be solved in~$O(4^nn^3)$ time using polynomial space.
\end{proposition}

\begin{proof}
	As shown by Bodlaender et al.~\cite{bodlaender2012note}, any vertex ordering problem can be solved in~$O(2^nn^{c+1})$ time and in~$O(4^nn^c)$ time using polynomial space if it can be reformulated as~$\min_{\pi} \sum_{v \in V} f(D,\pi_{<v},v)$, where~$\pi$ is a permutation of the vertices, $\pi_{<v}$ is the set of all vertices that appear before~$v$ in~$\pi$, and~$f$ can be computed in~$O(n^c)$ time.      
We show that \MinCost{} fits into this framework (with~$c=3$).
We only consider vertices in~$I$, that is, non-terminal vertices as these are all the vertices that should be eliminated.
We use the function \[f(D,\pi_{<v},v) = |\innbrs{D_{\pi_{<v}}}{v}| \cdot |\outnbrs{D_{\pi_{<v}}}{v}|.\]
Note that we can compute~$D_{\pi_{<v}}$ (and therefore~$f$) in~$O(n^3)$ time.
Note that given a permutation~$\pi$, the cost of eliminating all vertices in~$I$ exactly corresponds to~$\sum_{v \in V} f(D,\pi_{<v},v)$ as the cost of eliminating a vertex~$v$ in a solution sequence following~$\pi$ is exactly 
\({|\innbrs{D_{\pi_{<v}}}{v}| \cdot |\outnbrs{D_{\pi_{<v}}}{v}| = f(D,\pi_{<v},v)}.\)
This concludes the proof.
\end{proof}

\section{Conclusion}

We have resolved a pair of longstanding open questions by showing that \MinCost{} and \MinRep{} are both \NP-complete.
Our progress opens the door to many interesting questions.
On the theoretical side, a key next step is to understand the complexities of
both problems under the more expressive \emph{edge elimination} operation~\cite{naumann2002elimination}.
There are also promising opportunities to develop approximation algorithms and/or establish lower bounds.

\section*{Acknowledgments}
The authors would like to sincerely thank Paul Hovland for drawing their attention to the studied problems at Dagstuhl Seminar 24201, for insightful discussions, and for generously reviewing a preliminary version of this manuscript, providing valuable feedback and comments.

MB was supported by the European Research Council (ERC) project LOPRE
(819416) under the Horizon 2020 research and innovation program. AC, YM, and BDS were partially supported by the Gordon \& Betty Moore Foundation under grant GBMF4560 to BDS.

\bibliographystyle{abbrv}
\bibliography{references}
\end{document}